\documentclass{amsart}

\usepackage{amsmath}
\usepackage{amssymb,amsthm}
\usepackage{booktabs}
\usepackage{graphicx}
\usepackage{lmodern}
\usepackage{youngtab}
\usepackage{listings}
\usepackage{hyperref}

\newcommand{\ox}{\otimes}

\newcommand{\yfdet}{\textstyle\det_{d,n-d}^{\wedge 2}}
\newcommand{\rank}{\text{rank}}
\newcommand{\crank}{\text{krank}}
\newcommand{\wdots}{\wedge\ldots\wedge}
\newcommand{\intprod}{\reflectbox{\rotatebox[origin=c]{180}{$\neg$}}}

\DeclareMathOperator{\perm}{\mathrm{perm}}
\DeclareMathOperator{\im}{Im}

\newtheorem{theorem}{Theorem}[section]

\newtheorem{prop}[theorem]{Proposition}
\newtheorem{example}[theorem]{Example}
\newtheorem{lemma}[theorem]{Lemma}
\newtheorem{corollary}[theorem]{Corollary}

\theoremstyle{definition}
\newtheorem{definition}[theorem]{Definition}

\theoremstyle{remark}
\newtheorem{remark}[theorem]{Remark}

\title[Koszul-Young Flattenings and the Determinant]{Koszul-Young Flattenings and Symmetric Border Rank of the Determinant}
\author{Cameron Farnsworth}

\begin{document}

\begin{abstract}
 We present new lower bounds for the symmetric border rank of the $n\times n$ determinant for all $n$.  Further lower bounds are given for the $3\times 3$ permanent.
\end{abstract}

\maketitle

\section{Introduction}

The determinant polynomial is ubiquitous, its properties have been extensively studied. However basic questions regarding its complexity are still 
not understood.  Lower bounds for the (symmetric) border rank of a polynomial 
provide a measurement of its complexity and, as such, has become an area of growing 
interest.  In this paper we use techniques developed in \cite{LandOtt} to explore this question.  We prove a new lower bound for the 
symmetric border rank of the $n\times n$ determinant.

\begin{definition}
Let $V$ be a vector space and let $S^dV$ denote homogeneous degree $d$ polynomials on $V^\ast$. Given $P\in S^dV$, define its symmetric rank $R_s(P)$ by
\[
R_s(P)=\min\left\{r\in\mathbb{N}:P=\sum_{i=1}^r(v_i)^d,v_i\in V\right\}.
\]
\end{definition}
Symmetric rank is not semi-continuous under taking limits or Zariski closure, so we introduce symmetric border rank.
\begin{definition} \label{brankdef}
Let $P\in S^dV$.  Define the symmetric border rank of $P$, $\underline{R}_s(P)$ 
to be
\[
\underline{R}_s(P)=\min\left\{r\in\mathbb{N}:P\in\overline{\{T:R_s(T)=r\}}\right\}
\]where the overline denotes Zariski closure.
\end{definition}

\begin{theorem}

For $n\geq 5$, the following are lower bounds on the symmetric border rank of 
the determinant, $\underline{R}_s(\textstyle\det_n)$.

For $n$ even:
\[
\underline{R}_s(\textstyle\det_n)\geq\left(1+\frac{8(-8+6n^2+n^3)}{
(-1+n)(2+n)(4+n)^2(-2+n^2)}\right)\displaystyle\binom{n}{\frac{n}{2}}^2.
\]

For $n$ odd:
\[
\underline{R}_s(\textstyle\det_n)\geq\left(1+\frac{16(9+8n+n^2)}{
(3+n)(5+n)^2(-2+n^2)}\right)\displaystyle\binom{n}{\frac{n-1}{2}}^2.
\]
\end{theorem}

\begin{remark}
 Previously known lower bounds were
 \[
 \underline{R}_s(\textstyle\det_n)\geq\binom{n}{\frac{n}{2}}^2
 \]
 for $n$ even, and
 \[
 \underline{R}_s(\textstyle\det_n)\geq\binom{n}{\frac{n-1}{2}}^2
 \]
 for $n$ odd.
\end{remark}

\begin{remark}
 Asymptotically, our bound is
 \[
 \underline{R}_s(\textstyle\det_n)\gtrsim \frac{2^{2n+1}}{\pi\cdot n}+\frac{2^{2n+1}}{\pi\cdot n^4}
 \]
 whereas the previous lower bounds are approximately $\underline{R}_s(\textstyle\det_n)\gtrsim \frac{2^{2n+1}}{\pi\cdot n}$.
\end{remark}

\begin{theorem}
 $\underline{R}_s(\textstyle\det_4)\geq 38$.
\end{theorem}

\begin{remark}
 The previous bound was $\underline{R}_s(\textstyle\det_4)\geq 36$.
\end{remark}

\noindent Using a Macaulay2 \cite{M2} package developed by Steven Sam \cite{Sam}, we also show

\begin{theorem}
	\[
	\underline{R}_s(\textstyle\det_3)\geq 14
	\]
and
	\[
	\underline{R}_s(\textstyle\perm_3)\geq 14.
	\]
\end{theorem}

\begin{remark}
 The previous bounds were
	\[
	\underline{R}_s(\textstyle\det_3)\geq 9
	\]
and
	\[
	\underline{R}_s(\textstyle\perm_3)\geq 9.
	\]
\end{remark}

\begin{definition}
 Let $P\in S^dV$.  We define the Chow rank of $P$, $\rank_{Chow}(P)$, as
 \[
  \rank_{Chow}(P)=\min\{k : P=\sum_{i=1}^k\ell_{i1}\ldots\ell_{id} \ | \ \ell_{ij}\in V\}.
 \]
\end{definition}

In \cite{TeitlerIlten} it is shown that $\rank_{Chow}(\perm_3)=4$. Prior to this it was known that $\rank_{Chow}(\perm_3)\leq 4$ \cite{RyserFormula,GlynnPermSquare}.  Given $\rank_{Chow}(\perm_3)=4$, results from \cite{CarlCatGera} and \cite{BuczBuczTeit} proving $\underline{R}_s(x_1\cdots x_d)\leq 2^{d-1}$ show $\underline{R}_s(\perm_3)\leq 16$. In summary:

\begin{corollary}
 $14\leq\underline{R}_s(\perm_3)\leq 16$.
\end{corollary}

We may compare these lower bounds with known bounds on other ranks. $R_s(\det_n)\geq 
\binom{n}{\lfloor\frac{n}{2}\rfloor}^2+n^2-(\lfloor\frac{n}{2}\rfloor + 1)^2$ 
shown in \cite{LandTeit}, and for cactus rank, $\crank(\det_n)\geq\binom{2n}{n}-\binom{2n-2}{n-1}$ 
shown in \cite{DerkTeit} and before this it was known that 
$\crank(\det_n)\geq\frac{1}{2}\binom{2n}{n}$ \cite{Shafiei}. Known upper bounds for symmetric rank of $\det_n$ are $R_s(\det_n)\leq\left(\frac{5}{6}\right)^{\lfloor n/3\rfloor}2^{n-1}n!$ \cite{DerkNNorm} which also serves as an upper bound for symmetric border rank since $\underline{R}_s(T)\leq R_s(T)$ for any symmetric tensor.\\

\section{Background}

\indent Throughout this paper Young flattenings, a tool developed and used by Landsberg and Ottaviani \cite{LandOtt}, will be used extensively.  The irreducible polynomial representations of the general linear group, $GL(V)$, are parametrized by partitions $\pi$, where $\pi$ 
has at most $\dim V$ parts, see, e.g. \cite{FultYT,FultHarRep}. It is helpful 
to record these partitions visually by Young diagrams, which are 
left aligned diagrams consisting of boxes such that the $i$th row of the diagram has $\pi_i$ 
many boxes.
\begin{example}
The Young diagram corresponding to the partition $(4,3,1)$ of $8$ is
\[
\yng(4,3,1)
\]
\end{example}
 
\begin{prop}
[Pieri Formula, see e.g. \cite{FultHarRep,LandTensor}] Let $S_{\pi}V$ be an irreducible 
representation of $GL(V)$.  Then as a $GL(V)$-module
\[
S_{\pi}V\otimes S_{(d)}V=\bigoplus_{\substack{\mu \\ \ell(\mu)\leq\dim 
V}}S_{\mu}V
\]
where the partitions $\mu$ are obtained by adding $d$ boxes to $\pi$ so that no 
two boxes are added to the same column.
\end{prop}

\begin{definition}

Let partitions $\lambda$ and $\mu$ be such that $S_\mu V\subset S_\lambda V\ox 
S_{(d)} V$.  Given $P\in S_{(d)} V$ we 
obtain a linear map $\mathcal{F}_{\lambda,\mu}(P): S_\lambda V\to S_\mu V$ called a {\bf Young Flattening} via 
projecting the Pieri product $S_{\lambda} V\ox P$ to $S_\mu V$.

\end{definition}

\begin{prop}
[Proposition 4.1 of \cite{LandOtt}]\label{LOprop} Let $[x^d]\in v_d(\mathbb{P}V)$ and assume 
that $\rank(\mathcal{F}_{\lambda,\mu}(x^d))=t$.  If 
$\underline{R}_s(P)\leq r$, then 
$\rank(\mathcal{F}_{\lambda,\mu}(P))\leq rt$.
\end{prop}

\section{A Preliminary Result}

A preliminary result will be presented to make the method clear and prove the bound in the case $n=4$. By Proposition 2.6, to find a high lower bound for $\underline{R}_s(\det_n)$, we need to define a flattening such that $\rank(\mathcal{F}(\det_n))$ is big and $\rank(\mathcal{F}(x^n))$ is small.  Given $n$ dimensional vector spaces $A$ and $B$, and $\alpha\in S^d(A\otimes B)^\ast$, we will write $\alpha\raisebox{1pt}{\intprod}\det_n$ to denote the tensor contraction of $\alpha$ and $\det_n$.

\begin{remark}
If $\alpha$ is a minor of the determinant in the dual space $(A\otimes B)^\ast$, then $\alpha\raisebox{1pt}{\intprod}\det_n$ is a minor on the complementary indices in the primal space.
\end{remark}

For a tensor $\beta\in S^{n-d}(A\otimes B)$, let $\widehat{\beta}\in (A\otimes B)\otimes S^{n-d-1}(A\otimes B)$ be the image of $\beta$ under partial polarization.  Let $X^i_j:=a_i\otimes b_j$ and for $I,J\subset [n]$ with $|I|=|J|=n-d$, let $\Delta^I_J$ denote the $(n-d)\times (n-d)$ minor on the indices in $I$ and $J$.

\begin{remark}
$\widehat{\Delta}^I_J=\textstyle\sum\limits_{\substack{i\in I \\ j\in J}}(-1)^{i+j}X^i_j\otimes \Delta^{I\smallsetminus\{i\}}_{J\smallsetminus\{j\}}$
\end{remark}

\begin{remark}
The ``standard'' flattening of the determinant is $\det_{d,n-d}:S^d(A\otimes B)^\ast\longrightarrow S^{n-d}(A\otimes B)$ defined by $\alpha\mapsto \alpha\raisebox{1pt}{\intprod}\det_n$. Then $\im (\det_{d,n-d})$ is spanned by the $(n-d)\times(n-d)$ minors of the determinant.
\end{remark}

Define the Young flattening
\begin{align*}
\textstyle\det_{d,n-d}^{\wedge 1}:\bigwedge^{n-d}A\otimes\bigwedge^{n-d}B\otimes(A\otimes B)&\longrightarrow\textstyle\bigwedge^{n-d-1}A\otimes\bigwedge^{n-d-1}B\otimes\bigwedge^2(A\otimes B)\\
\Delta^I_J\otimes v&\mapsto \textstyle\sum\limits_{\substack{i\in I \\ j\in J}}(-1)^{i+j}X^i_j\wedge v\otimes \Delta^{[n-d]\smallsetminus\{i\}}_{[n-d]\smallsetminus\{j\}}
\end{align*}
and extend linearly.

\begin{lemma}\label{prelimImageDecomp}
 $\im(\textstyle\det_{d,n-d}^{\wedge 1})$ is contained in\\
 \begin{align*}
  &S_{2,1^{n-d-1}}A\ox S_{1^{n-d+1}}B\oplus S_{1^{n-d+1}}A\ox S_{2,1^{n-d-1}}B \tag{*} \\
  \oplus&S_{2,1^{n-d-1}}A\ox S_{2,1^{n-d-1}}B.
 \end{align*}
\end{lemma}

\begin{proof}
 Decomposing $\bigwedge^{n-d}A\ox\bigwedge^{n-d}B\ox A\ox B$ as a $GL_n\times GL_n$-module we get
 
 \begin{align*}
 &S_{2,1^{n-d-1}}A\ox S_{1^{n-d+1}}B\oplus S_{1^{n-d+1}}A\ox S_{2,1^{n-d-1}}B \\
  \oplus&S_{2,1^{n-d-1}}A\ox S_{2,1^{n-d-1}}B\oplus S_{1^{n-d+1}}A\ox S_{1^{n-d+1}}
 \end{align*}
 
 and $\bigwedge^{n-d-1}A\ox\bigwedge^{n-d-1}B\ox\bigwedge^2(A\ox B)$ as $GL_n\times GL_n$-module decomposes as
 
 \begin{align*}
 &S_{1^{n-d+1}}A\ox S_{3,1^{n-d-2}}B\oplus S_{2,1^{n-d-1}}A\ox S_{3,1^{n-d-2}}B\\
 \oplus&S_{2,2,1^{n-d-3}}A\ox S_{3,1^{n-d-2}}B\oplus S_{1^{n-d+1}}A\ox S_{2,1^{n-d-1}}B\\
 \oplus&(S_{2,1^{n-d-1}}A\ox S_{2,1^{n-d-1}}B)^{\oplus 2}\oplus S_{2,2,1^{n-d-3}}A\ox S_{2,1^{n-d-1}}B\\
 \oplus&S_{3,1^{n-d-2}}A\ox S_{1^{n-d+1}}B\oplus S_{3,1^{n-d-2}}A\ox S_{2,1^{n-d-1}}B\\
 \oplus&S_{3,1^{n-d-2}}A\ox S_{2,2,1^{n-d-3}}B\oplus S_{2,1^{n-d-1}}A\ox S_{1^{n-d+1}}B\\
 \oplus&S_{2,1^{n-d-1}}A\ox S_{2,2,1^{n-d-3}}B
 \end{align*}
 
 The irreducible modules in Lemma \ref{prelimImageDecomp} are the only irreducible modules appearing in both decompositions.  By Schur's lemma, we conclude that the module $(*)$ must contain $\im(\textstyle\det_{d,n-d}^{\wedge 1})$.  
\end{proof}

It must now be verified for each irreducible module in $(*)$, that $\textstyle\det_{d,n-d}^{\wedge 1}$ is not the zero map on the module. Since each irreducible module appears with multiplicity 1, then for a given irreducible module with highest weight $\pi$, finding any highest weight vector $v\in\bigwedge^{n-d}A\ox\bigwedge^{n-d}B\ox(A\ox B)$ of weight $\pi$ such that $\textstyle\det_{d,n-d}^{\wedge 1}(v)\neq 0$ proves $\textstyle\det_{d,n-d}^{\wedge 1}$ is nonzero on the entire module.

\begin{lemma}
 $\textstyle\det_{d,n-d}^{\wedge 1}$ is an isomorphism on the irreducible module $S_{2,1^{n-d-1}}A\ox S_{2,1^{n-d-1}}B$.
\end{lemma}

\begin{proof}
 Consider $a_1\wedge\ldots\wedge a_{n-d}\otimes a_1\otimes b_1\wedge\ldots\wedge b_{n-d}\otimes b_1$, a highest weight vector of the irreducible module $S_{2,1^{n-d-1}}A\ox S_{2,1^{n-d-1}}B$. Its projection into $(A\ox B)\ox \bigwedge^{n-d}A\ox\bigwedge^{n-d}B$ is a nonzero multiple of
 \[
 X^1_1\ox\Delta^{[n-d]}_{[n-d]}.
 \]
 Then
 \begin{align*}
  &\textstyle\det_{d,n-d}^{\wedge 1}(X^1_1\ox\Delta^{[n-d]}_{[n-d]})\\
  &=\sum_{\substack{i\in[n-d]\\j\in[n-d]}}(-1)^{i+j}X^1_1\wedge X^i_j\otimes\Delta^{[n-d]\smallsetminus\{i\}}_{[n-d]\smallsetminus\{j\}}.
 \end{align*}
Note that the term $X^1_1\wedge X^1_2\otimes\Delta^{[n-d]\smallsetminus\{1\}}_{[n-d]\smallsetminus\{2\}}$ will not cancel in the sum.
\end{proof}

\begin{lemma}
 $\textstyle\det_{d,n-d}^{\wedge 1}$ is an isomorphism on the irreducible modules $S_{2,1^{n-d-1}}A\ox S_{1^{n-d+1}}B$ and by symmetry $S_{1^{n-d+1}}A\ox S_{2,1^{n-d-1}}B$ is not in the kernel.
\end{lemma}

\begin{proof}
 Consider $a_1\wedge\ldots\wedge a_{n-d}\otimes a_1\otimes b_1\wedge\ldots\wedge b_{n-d+1}$, a highest weight vector of the irreducible module $S_{2,1^{n-d-1}}A\ox S_{1^{n-d+1}}B$. Its projection into $(A\ox B)\ox \bigwedge^{n-d}A\ox\bigwedge^{n-d}B$ is a nonzero multiple of
 \[
 \sum_{j\in[n-d+1]}(-1)^{j}X^1_j\otimes\Delta^{[n-d]}_{[n-d+1]\smallsetminus\{j\}}.
 \]
 Then
 \begin{align*}
  &\textstyle\det_{d,n-d}^{\wedge 1}(\sum_{j\in[n-d+1]}(-1)^{j}X^1_j\otimes\Delta^{[n-d]}_{[n-d+1]\smallsetminus\{j\}})\\
  &=\sum_{j\in[n-d+1]}\sum_{\substack{i\in[n-d]\\k\in[n-d+1]\smallsetminus\{j\}}}(-1)^{j}(-1)^{i+\tilde{k}}X^1_j\wedge X^i_k\otimes\Delta^{[n-d]\smallsetminus\{i\}}_{[n-d+1]\smallsetminus\{j,k\}}
 \end{align*}
where
\[
\tilde{k}:=
\begin{cases}k,&k<j\\k-1,&j<k.
\end{cases}\]
Note that $X^1_1\wedge X^1_2\otimes\Delta^{[n-d]\smallsetminus\{1\}}_{[n-d+1]\smallsetminus\{1,2\}}$ does not cancel in the sum.
\end{proof}

Finding a value of $d$ with respect to $n$ that maximizes the rank of $\textstyle\det_{d,n-d}^{\wedge 1}$ and dividing by the rank of $[x^n]^{\wedge 1}_{d,n-d}$ we demonstrate the following theorem.

\begin{theorem}
 For $n\geq 3$, the following are lower bounds on the symmetric border rank of the determinant, $\underline{R}_s(\det_n)$.\\
 For $n$ even:
 \[
 \underline{R}_s(\det\nolimits_n)\geq\left(1+\frac{4}{(-1+n)(2+n)^2}\right)\binom{n}{\frac{n}{2}}^2
 \]
 For $n$ odd:
 \[
 \underline{R}_s(\det\nolimits_n)\geq\left(1+\frac{8}{(-1+n)(3+n)^2}\right)\binom{n}{\frac{n-1}{2}}^2.
 \]
\end{theorem}

\section{Proof of Main Theorem}

To prove the main theorem, we use the map
\[
 \yfdet:\bigwedge^{n-d}A\otimes\bigwedge^{n-d}B\otimes\bigwedge^2(A\otimes B)\longrightarrow\bigwedge^{n-d-1}A\otimes\bigwedge^{n-d-1}B\otimes\bigwedge^3(A\otimes B)
\]
defined by
\[
\Delta^I_J\otimes v\wedge w\mapsto \textstyle\sum\limits_{\substack{i\in I \\ j\in J}}(-1)^{i+j}X^i_j\wedge v\wedge w\otimes \Delta^{[n-d]\smallsetminus\{i\}}_{[n-d]\smallsetminus\{j\}}
\]
and extended linearly. It remains to find the rank of $\yfdet$.

\begin{lemma}\label{lemmaDecomp}
$\im(\yfdet)$ is contained in
\begin{align*}
&S_{3,1^{n-d-1}}A\otimes S_{1^{n-d+2}}B\oplus S_{1^{n-d+2}}A\otimes S_{3,1^{n-d-1}}B\oplus S_{3,1^{n-d-1}}A\otimes S_{2,1^{n-d}}B\\
\oplus &S_{2,1^{n-d}}A\otimes S_{3,1^{n-d-1}}B\oplus S_{3,1^{n-d-1}}A\otimes S_{2,2,1^{n-d-2}}B\\
\oplus &S_{2,2,1^{n-d-2}}A\otimes S_{3,1^{n-d-1}}B\oplus S_{2,1^{n-d+1}}A\otimes S_{2,1^{n-d+1}}B\\
\oplus &S_{2,1^{n-d+1}}A\otimes S_{2,2,1^{n-d-1}}B\oplus S_{2,2,1^{n-d-1}}A\otimes S_{2,1^{n-d+1}}B
\end{align*}
\end{lemma}

\begin{proof}
Decomposing $\bigwedge^{n-d}A\otimes\bigwedge^{n-d}B\otimes\bigwedge^2(A\otimes B)$ and $\bigwedge^{n-d-1}A\otimes\bigwedge^{n-d-1}B\otimes\bigwedge^3(A\otimes B)$ as $GL_n\times GL_n$-modules, one sees that only the irreducibles listed in the lemma appear in both decompositions and that the minimum multiplicity each appears with is $1$.  By Schur's Lemma, no other irreducible may be in the image.
\end{proof}

The above lemma gives us an idea as to the largest lower bound that this particular flattening could achieve. However, we are not guaranteed that this is the image.  To proceed, for each irreducible module in the lemma we must find a highest weight vector and compute $\yfdet$ on this vector.  Note since each module appears with multiplicity $1$, finding a single highest weight vector of the correct highest weight on which the flattening is nonzero is sufficient.\\

\begin{lemma}
 $\yfdet$ is an isomorphism on the irreducible module $S_{3,1^{n-d-1}}A\otimes S_{1^{n-d+2}}B$ and by symmetry on $S_{1^{n-d+2}}A\otimes S_{3,1^{n-d-1}}B$.
\end{lemma}

\begin{proof}
 Consider $a_1\wdots a_{n-d}\ox a_1\ox a_1\ox b_1\wdots b_{n-d+2}$, a highest weight vector of the irreducible module $S_{3,1^{n-d-1}}A\otimes S_{1^{n-d+2}}B$. Its projection into $\bigwedge^{2}(A\ox B)\ox\bigwedge^{n-d}A\ox\bigwedge^{n-d}B$ is a multiple of
 \[
\textstyle\sum\limits_{1\leq i<j\leq n-d+2}(-1)^{i+j}X^1_i\wedge X^1_j\ox\Delta^{[n-d]}_{[n-d+2]\smallsetminus\{i,j\}}.
\]
Then
\begin{align*}
&\yfdet\left(\textstyle\sum\limits_{1\leq i<j\leq n-d+2}(-1)^{i+j}X^1_i\wedge X^1_j\ox\Delta^{[n-d]}_{[n-d+2]\smallsetminus\{i,j\}}\right)\\
=\textstyle\sum\limits_{1\leq i<j\leq n-d+2}&\left(\textstyle\sum\limits_{h=1}^{n-d}\textstyle\sum\limits_{k\in[n-d+2]\smallsetminus\{i,j\}}(-1)^{\tilde{k}+h}(-1)^{i+j}X^1_i\wedge X^1_j\wedge X^h_k\ox\Delta^{[n-d]\smallsetminus\{h\}}_{[n-d+2]\smallsetminus\{i,j,k\}}\right)
\end{align*}
where \[
\tilde{k}:=
\begin{cases} k,&k<i<j\\ k-1,&i<k<j\\k-2,&i<j<k
.\end{cases}\]
Then note that the term $X^1_1\wedge X^1_2\wedge X^1_3\ox\Delta^{[n-d]\smallsetminus\{1\}}_{[n-d+2]\smallsetminus\{1,2,3\}}$ does not cancel in the sum.
\end{proof}

\begin{lemma}
 $\yfdet$ is an isomorphism on the irreducible module $S_{3,1^{n-d-1}}A\otimes S_{2,1^{n-d}}B$ and by symmetry on $S_{2,1^{n-d}}A\otimes S_{3,1^{n-d-1}}B$.
\end{lemma}

\begin{proof}
 Consider $a_1\wdots a_{n-d}\ox a_1\ox a_1\ox b_1\wdots b_{n-d+1}\ox b_1$, a highest weight vector of the irreducible module $S_{3,1^{n-d-1}}A\otimes S_{2,1^{n-d}}B$. Its projection to $\bigwedge^{2}(A\ox B)\ox\bigwedge^{n-d}A\ox\bigwedge^{n-d}B$ is a multiple of
\[
\textstyle\sum\limits_{i=2}^{n-d+1}(-1)^{i}X^1_1\wedge X^1_i\ox\Delta^{[n-d]}_{[n-d+1]\smallsetminus\{i\}}.
\]
Then
\begin{align*}
 &\yfdet\left(\textstyle\sum\limits_{i=2}^{n-d+1}(-1)^{i}X^1_1\wedge X^1_i\ox\Delta^{[n-d]}_{[n-d+1]\smallsetminus\{i\}}\right)\\
 &=\textstyle\sum\limits_{k=1}^{n-d}\textstyle\sum\limits_{i=2}^{n-d+1}
\textstyle\sum\limits_{j\in[n-d+1]\smallsetminus\{i\}}(-1)^{i}(-1)^{\tilde{j}+k}X^1_1\wedge X^1_i\wedge X^k_j\ox\Delta^{[n-d]\smallsetminus\{k\}}_{[n-d+1]\smallsetminus\{i,j\}}
\end{align*}
where \[
\tilde{j}:=
\begin{cases}j,&j<i\\j-1,&i<j
.\end{cases}\]
The observation that $X^1_1\wedge X^1_3\wedge X^1_2\ox\Delta^{[n-d]\smallsetminus\{1\}}_{[n-d+1]\smallsetminus\{2,3\}}$ does not cancel demonstrates the lemma.
\end{proof}

\begin{lemma}
 $\yfdet$ is an isomorphism on the irreducible module $S_{3,1^{n-d-1}}A\otimes S_{2,2,1^{n-d-2}}B$ and by symmetry on $S_{2,2,1^{n-d-2}}A\otimes S_{3,1^{n-d-1}}B$.
\end{lemma}

\begin{proof}
 Consider $a_1\wdots a_{n-d}\ox a_1\ox a_1\ox b_1\wdots b_{n-d}\ox b_1\wedge b_2$, a highest weight vector of the irreducible module $S_{3,1^{n-d-1}}A\otimes S_{2,2,1^{n-d-2}}B$. Its projection to $\bigwedge^{2}(A\ox B)\ox\bigwedge^{n-d}A\ox\bigwedge^{n-d}B$ is a multiple of
 \[
 X_1^1\wedge X^1_2\ox\Delta^{[n-d]}_{[n-d]}.
 \]
 Then
 \[
 \yfdet\left(X_1^1\wedge X^1_2\ox\Delta^{[n-d]}_{[n-d]}\right)=\textstyle\sum\limits_{i,j=1}^{n-d}(-1)^{j+i}X_1^1\wedge X^1_2\wedge X^i_j\ox\Delta^{[n-d]\smallsetminus\{i\}}_{[n-d]\smallsetminus \{j\}}.
 \]
 We may see this is not zero since the term $X^1_1\wedge X^1_2\wedge X^1_3\ox\Delta^{[n-d]\smallsetminus\{1\}}_{[n-d]\smallsetminus\{3\}}$ appears in the sum only once.
\end{proof}

\begin{lemma}
 $\yfdet$ is an isomorphism on the irreducible module $S_{2,1^{n-d}}A\otimes S_{2,1^{n-d}}B$.
\end{lemma}

\begin{proof}
Consider $a_1\wdots a_{n-d+1}\ox a_1\ox b_1\wdots b_{n-d+1}\ox b_1$, a highest weight vector of the irreducible module $S_{2,1^{n-d}}A\otimes S_{2,1^{n-d}}B$. Its projection to $\bigwedge^{2}(A\ox B)\ox\bigwedge^{n-d}A\ox\bigwedge^{n-d}B$ is a multiple of
\[
\textstyle\sum\limits_{i=1}^{n-d+1}\textstyle\sum\limits_{j=2}^{n-d+1}(-1)^{i+j}X^1_1\wedge X^i_j\ox\Delta^{[n-d+1]\smallsetminus\{i\}}_{[n-d+1]\smallsetminus\{j\}}+\textstyle\sum\limits_{i=1}^{n-d+1}\textstyle\sum\limits_{j=2}^{n-d+1}(-1)^{i+j}X^i_1\wedge X^1_j\ox\Delta^{[n-d+1]\smallsetminus\{i\}}_{[n-d+1]\smallsetminus\{j\}}.
\]
Then
\begin{align*}
 \yfdet\left(\textstyle\sum\limits_{i=1}^{n-d+1}\textstyle\sum\limits_{j=2}^{n-d+1}\right.&\left.(-1)^{i+j}X^1_1\wedge X^i_j\ox\Delta^{[n-d+1]\smallsetminus\{i\}}_{[n-d+1]\smallsetminus\{j\}}\right.\\
 &\left.+\textstyle\sum\limits_{i=1}^{n-d+1}\textstyle\sum\limits_{j=2}^{n-d+1}(-1)^{i+j}X^i_1\wedge X^1_j\ox\Delta^{[n-d+1]\smallsetminus\{i\}}_{[n-d+1]\smallsetminus\{j\}}\right)\\
 =\textstyle\sum\limits_{i=1}^{n-d+1}\textstyle\sum\limits_{j=2}^{n-d+1}\textstyle\sum\limits_{\substack{k\in[n-d+1]\smallsetminus\{i\}\\l\in[n-d+1]\smallsetminus\{j\}}}&(-1)^{i+j}(-1)^{\tilde{k}+\tilde{l}}X^1_1\wedge X^i_j\wedge X^k_l\ox\Delta^{[n-d+1]\smallsetminus\{i,k\}}_{[n-d+1]\smallsetminus\{j,l\}}\\
+\textstyle\sum\limits_{i=1}^{n-d+1}\textstyle\sum\limits_{j=2}^{n-d+1}\textstyle\sum\limits_{\substack{k\in[n-d+1]\smallsetminus\{i\}\\l\in[n-d+1]\smallsetminus\{j\}}}&(-1)^{i+j}(-1)^{\tilde{k}+\tilde{l}}X^i_1\wedge X^1_j\wedge X^k_l\ox\Delta^{[n-d+1]\smallsetminus\{i,k\}}_{[n-d+1]\smallsetminus\{j,l\}}
\end{align*}
where
\[
\tilde{k}:=
\begin{cases}k,&k<i\\k-1,&i<k
\end{cases}\]
and \[
\tilde{l}:=
\begin{cases}l,&l<j\\l-1,&j<l
.\end{cases}
\]
Since $X^1_1\wedge X^1_2\wedge X^2_1\ox\Delta^{[n-d+1]\smallsetminus\{1,2\}}_{[n-d+1]\smallsetminus\{2,1\}}$ does not cancel the lemma is proven.
\end{proof}

\begin{lemma}
 $\yfdet$ is an isomorphism on the irreducible module $S_{2,2,1^{n-d-2}}A\otimes S_{2,1^{n-d}}B$ and by symmetry on $S_{2,1^{n-d}}A\otimes S_{2,2,1^{n-d-2}}B$.
\end{lemma}

\begin{proof}
 Consider $a_1\wdots a_{n-d}\ox a_1\wedge a_2\ox b_1\wdots b_{n-d+1}\ox b_1$, a highest weight vector of the irreducible module $S_{2,2,1^{n-d-2}}A\otimes S_{2,1^{n-d}}B$. Its projection to $\bigwedge^{2}(A\ox B)\ox\bigwedge^{n-d}A\ox\bigwedge^{n-d}B$ is a multiple of
 \[
\textstyle\sum\limits_{i=1}^{n-d+1}(-1)^{i}X^1_1\wedge 
X^2_i\ox\Delta^{[n-d]}_{[n-d+1]\smallsetminus\{i\}}+\textstyle\sum\limits_{i=1}^
{n-d+1}(-1)^{i}X^1_i\wedge X^2_1\ox\Delta^{[n-d]}_{[n-d+1]\smallsetminus\{i\}}
\]
Then
\begin{align*}
 &\yfdet\left(\textstyle\sum\limits_{i=1}^{n-d+1}(-1)^{i}X^1_1\wedge X^2_i\ox\Delta^{[n-d]}_{[n-d+1]\smallsetminus\{i\}}+\textstyle\sum\limits_{i=1}^{n-d+1}(-1)^{i}X^1_i\wedge X^2_1\ox\Delta^{[n-d]}_{[n-d+1]\smallsetminus\{i\}}\right)\\
 &=\textstyle\sum\limits_{k=1}^{n-d}\textstyle\sum\limits_{i=1}^{n-d+1}\textstyle\sum\limits_{j\in[n-d+1]\smallsetminus\{i\}}(-1)^{i}(-1)^{\tilde{j}+k}X^1_1\wedge X^2_i\wedge X^k_j\ox\Delta^{[n-d]\smallsetminus\{k\}}_{[n-d+1]\smallsetminus\{i,j\}}\\
 &+\textstyle\sum\limits_{k=1}^{n-d}\textstyle\sum\limits_{i=1}^{n-d+1}\textstyle\sum\limits_{j\in[n-d+1]\smallsetminus\{i\}}(-1)^{i}(-1)^{\tilde{j}+k}X^1_i\wedge X^2_1\wedge X^k_j\ox\Delta^{[n-d]\smallsetminus\{k\}}_{[n-d+1]\smallsetminus\{i,j\}}
\end{align*}
where \[
\tilde{j}:=
\begin{cases}j,&j<i\\j-1,&i<j
.\end{cases}\]
Observing that $X^1_1\wedge X^2_1\wedge X^1_2\ox\Delta^{[n-d]\smallsetminus\{1\}}_{[n-d+1]\smallsetminus\{1,2\}}$ does not cancel proves the lemma.
\end{proof}

\begin{lemma}\label{lemmaImage}
 The image of $\yfdet$ consists of all of the irreducible modules in the decomposition in Lemma \ref{lemmaDecomp}.
\end{lemma}

\begin{proof}
 This is demonstrated by the preceding lemmas.
\end{proof}

\begin{lemma}\label{lemmaImDim}
 $\dim(\im(\yfdet))$ has a maximum at $d=\lfloor\frac{n}{2}\rfloor$.
\end{lemma}

\begin{proof}
Begin by factoring $\dim(\im(\yfdet))$ into the form $f(n,d)\binom{n}{d}^2$, where $f(n,d)$ is a rational function of $n$ and $d$. In particular

\begin{align*}
f(n,d)&:=\frac{(n+2)(n+1)(n-d)(d)(d-1)}{(n-d+2)^2(n-d+1)}+\frac{
(n+2)(n+1)^2(n-d)(d)}{(n-d+2)^2}\\
&+\frac{(n+2)(n+1)^2(n-d)(n)(n-d-1)}{2(n-d+2)(n-d+1)}+\frac{(n+1)^2(n)(n-d-1)(d)
}{(n-d+1)(n-d+2)}\\
&+\frac{(n+1)^2(d)^2}{(n-d+2)^2}
\end{align*}

Then consider
\[
f(n,d)\binom{n}{d}^2-f(n,d+1)\binom{n}{d+1}^2
\]
and rewrite it as
\[
 \left( f(n,d)-f(n,d+1)\frac{(n-d)^2}{(d+1)^2}\right)\binom{n}{d}^2.
\]
Notice that $f(n,d)-f(n,d+1)\frac{(n-d)^2}{(d+1)^2}<0$ for $d=\lfloor\frac{n}{2}\rfloor-1$ and $f(n,d)-f(n,d+1)\frac{(n-d)^2}{(d+1)^2}>0$ for $d=\lfloor\frac{n}{2}\rfloor$ and conclude the lemma.
\end{proof}

\begin{remark}
 The requirement for $n\geq 5$ in the main theorem, is so that the length of all partitions $S_{1^{n-d+2}}A$, $S_{2,1^{n-d}}A$, $S_{3,1^{n-d-1}}A$, and $S_{2,2,1^{n-d-2}}A$ (respectively $B$) do not exceed $\dim(A)=\dim(B)=n$.  Hence, all of the irreducible modules in the decomposition in Lemma \ref{lemmaDecomp} occur when $d=\lfloor\frac{n}{2}\rfloor$.
\end{remark}

\begin{remark}\label{remVerRank}
$\rank([{(X^i_j)}^n]^{\wedge 2}_{d,n-d})=\binom{n^2-1}{2}$.  Which may be verified easily by noticing the image of contracting $\alpha\in S^d(A\otimes B)*$ with $(X^i_j)^n$ is in the span of $(X^i_j)^{n-d}$ and $\widehat{(X^i_j)^{n-d}}$ is in the span of $(X^i_j)^{n-d-1}\otimes X^i_j$.  Hence $\im([{(X^i_j)}^n]^{\wedge 2}_{d,n-d}))$ are of the form $(X^i_j)^{n-d-1}\otimes X^i_j\wedge v\wedge w$, where $v$ and $w$ cannot be in the span of $X^i_j$.
\end{remark}

The main theorem follows by substituting $\lfloor\frac{n}{2}\rfloor$ into $f(n,d)$ from the proof of Lemma \ref{lemmaImDim}, dividing by $\binom{n^2-1}{2}$ which is the rank from Remark \ref{remVerRank}, and simplifying. 

\section{$3\times 3$ determinant and permanent}

Define the partitions $\pi_n=((n-1)^{n+1},(n-2)^{n+1},\ldots,1^{n+1})$ and 
$\tilde{\pi}_n=(n,\pi_n)$.  For example, $\pi_3=(2^4,1^4)$ and let 
$\tilde{\pi}_3=(3,2^4,1^4)$.  Note that 
$\dim(S_{\pi_3}\mathbb{C}^{9})=\dim(S_{\tilde{\pi}_3}\mathbb{C}^{9})=1050$. For 
a polynomial $\phi\in S^3\mathbb{C}^9$, define the Young flattening 
	\[
	\mathcal{F}_{\pi_3,\tilde{\pi}_3}(\phi): S_{\pi_3}\mathbb{C}^9\to 
S_{\tilde{\pi}_3}\mathbb{C}^9
	\]
by the labeled Pieri product restricted to shape $\tilde{\pi}_3$

	\[
	T_{\pi_3}\otimes\phi=\sum 
c_{T_{\pi_3},\tilde{T}_{\tilde{\pi}_3}}\tilde{T}_{\tilde{\pi}_3}
	\]
where $T_{\pi_3}$ and $\tilde{T}_{\tilde{\pi}_3}$ are semi-standard fillings of 
tableaux of shape $\pi_3$ and $\tilde{\pi}_3$ respectively and where 
$c_{T_{\pi_3},\tilde{T}_{\tilde{\pi}_3}}$ is obtained by adding boxes to 
$\pi_3$ such as to obtain a tableau of shape $\tilde{\pi}_3$ and for each 
monomial in $\phi$, label the boxes with the variable names in all 
permutations and straighten.  $c_{T_{\pi_3},\tilde{T}_{\tilde{\pi}_3}}$ is the 
coefficient of $\tilde{T}_{\tilde{\pi}_3}$.

Consider the polynomial $(x_{3,3})^3\in S^3\mathbb{C}^9$, we immediately 
see that if $T_{\pi_3}$ has any box labeled $x_{3,3}$, then 
$\mathcal{F}_{\pi_3,\tilde{\pi}_3}((x_{3,3})^3)=0$.  Since this is the only 
restriction of tableaux,

	\[
	\dim \mathrm{Im}(\mathcal{F}_{\pi_3,\tilde{\pi}_3}((x_{3,3})^3))=\dim 
S_{\pi_3}\mathbb{C}^8=70.
	\]
By proposition \ref{LOprop}, if 
$[x^3]\in v_3(\mathbb{P}\mathbb{C}^9)$ has $\rank\mathcal{F}_{\mu,\nu}(x^3)=p$, 
then for $[\phi]\in \mathbb{P}S^3\mathbb{C}^9$ with rank $r$, 
$\rank(\mathcal{F}_{\mu,\nu}(\phi))\leq rp$.
Thus the maximum lower bound on symmetric border rank on 
polynomial $\phi\in S^3\mathbb{C}^9$ this method may achieve is

	\[
	\underline{R}_s(\phi)\geq 15
	\]
This being when $\dim 
\mathrm{Im}(\mathcal{F}_{\pi_3,\tilde{\pi}_3}(\phi))=1050$.  Applying this 
flattening to $\det_3$ and $\perm_3$ and using the the Macaulay2 \cite{M2} package \verb PieriMaps \ developed by Steven Sam \cite{Sam} we get
	\[
	\dim \mathrm{Im}(\mathcal{F}_{\pi_3,\tilde{\pi}_3}(\textstyle\det _ 
3))=950
	\]
and
	\[
	\dim \mathrm{Im}(\mathcal{F}_{\pi_3,\tilde{\pi}_3}(\textstyle\perm _ 
3))=934.
	\]
These give the following lower bounds
	\[
	\underline{R}_s(\textstyle\det_3)\geq 14
	\]
and
	\[
	\underline{R}_s(\textstyle\perm_3)\geq 14.
	\]
This is an improvement from the classical lower bound for the determinant of 9 
and the bound obtained from the Koszul-Young flattening $\det_{1,2}^{\wedge 2}$ 
of 12.

The following code is used to complete the above computations.

\begin{verbatim}
loadPackage"PieriMaps"
A=QQ[x_(0,0)..x_(2,2)]
time MX = pieri({3,2,2,2,2,1,1,1,1},{1,5,9},A);
rank diff(x_(0,0)^3,MX)
f = det genericMatrix(A,x_(0,0), 3,3)
rank diff(f,MX)
g =x_(0,2)*x_(1,1)*x_(2,0)+x_(0,1)*x_(1,2)*x_(2,0)+
   x_(0,2)*x_(1,0)*x_(2,1)+x_(0,0)*x_(1,2)*x_(2,1)+
   x_(0,1)*x_(1,0)*x_(2,2)+x_(0,0)*x_(1,1)*x_(2,2)
rank diff(g,MX)
\end{verbatim}

\section*{Acknowledgement}

Part of this work was done while the author was visiting the Simons Institute for the Theory of Computing, UC Berkeley for the Algorithms and Complexity in Algebraic Geometry program.  The author would also like to thank Luke Oeding for useful conversation and help with Macaulay2 computations, Christian Ikenmeyer for useful conversation, and J.M. Landsberg for his guidance.

\bibliography{mybib}

\begin{thebibliography}{10}

\bibitem{BuczBuczTeit}
Weronika Buczy{\'n}ska, Jaros{\l}aw Buczy{\'n}ski, and Zach Teitler.
\newblock Waring decompositions of monomials.
\newblock {\em Journal of Algebra}, 378(0):45 -- 57, 2013.

\bibitem{CarlCatGera}
Enrico Carlini, Maria~Virginia Catalisano, and Anthony~V. Geramita.
\newblock The solution to the waring problem for monomials and the sum of
  coprime monomials.
\newblock {\em Journal of Algebra}, 370(0):5 -- 14, 2012.

\bibitem{DerkNNorm}
Harm {Derksen}.
\newblock {On the Nuclear Norm and the Singular Value Decomposition of
  Tensors}.
\newblock {\em ArXiv e-prints}, August 2013.

\bibitem{DerkTeit}
Harm Derksen and Zach Teitler.
\newblock Lower bound for ranks of invariant forms.
\newblock {\em J. Pure Appl. Algebra}, 219(12):5429--5441, 2015.

\bibitem{FultYT}
William Fulton.
\newblock {\em Young tableaux}, volume~35 of {\em London Mathematical Society
  Student Texts}.
\newblock Cambridge University Press, Cambridge, 1997.
\newblock With applications to representation theory and geometry.

\bibitem{FultHarRep}
William Fulton and Joe Harris.
\newblock {\em Representation theory}, volume 129 of {\em Graduate Texts in
  Mathematics}.
\newblock Springer-Verlag, New York, 1991.
\newblock A first course, Readings in Mathematics.

\bibitem{GlynnPermSquare}
David~G. Glynn.
\newblock The permanent of a square matrix.
\newblock {\em European J. Combin.}, 31(7):1887--1891, 2010.

\bibitem{M2}
Daniel~R. Grayson and Michael~E. Stillman.
\newblock Macaulay2, a software system for research in algebraic geometry.
\newblock Available at \url{http://www.math.uiuc.edu/Macaulay2/}.

\bibitem{TeitlerIlten}
Nathan {Ilten} and Zach {Teitler}.
\newblock {Product Ranks of the $3 \times 3$ Determinant and Permanent}.
\newblock {\em ArXiv e-prints}, March 2015.

\bibitem{LandTensor}
J.~M. Landsberg.
\newblock {\em Tensors: geometry and applications}, volume 128 of {\em Graduate
  Studies in Mathematics}.
\newblock American Mathematical Society, Providence, RI, 2012.

\bibitem{LandTeit}
J.~M. Landsberg and Zach Teitler.
\newblock On the ranks and border ranks of symmetric tensors.
\newblock {\em Found. Comput. Math.}, 10(3):339--366, 2010.

\bibitem{LandOtt}
J.M. Landsberg and Giorgio Ottaviani.
\newblock Equations for secant varieties of veronese and other varieties.
\newblock {\em Annali di Matematica Pura ed Applicata}, 192(4):569--606, 2013.

\bibitem{RyserFormula}
Herbert~John Ryser.
\newblock {\em Combinatorial mathematics}.
\newblock The Carus Mathematical Monographs, No. 14. Published by The
  Mathematical Association of America; distributed by John Wiley and Sons,
  Inc., New York, 1963.

\bibitem{Sam}
Steven~V. Sam.
\newblock Computing inclusions of {S}chur modules.
\newblock {\em J. Softw. Algebra Geom.}, 1:5--10, 2009.

\bibitem{Shafiei}
Sepideh~Masoumeh Shafiei.
\newblock Apolarity for determinants and permanents of generic matrices.
\newblock {\em J. Commut. Algebra}, 7(1):89--123, 2015.

\end{thebibliography}
\bibliographystyle{plain}

\end{document}